\documentclass[11pt]{article}
\usepackage{amssymb,amsmath}
\usepackage{epsfig}
\usepackage{hyperref}
\usepackage{cite} 

\usepackage{times}

\newcommand{\qedsymb}{\hfill{\rule{2mm}{2mm}}}
\newenvironment{proof}[1][]{\begin{trivlist}
\item[\hspace{\labelsep}{\bf\noindent Proof#1:\/}] }{\qedsymb\end{trivlist}}

\mathchardef\ordinarycolon\mathcode`\:     
\mathcode`\:=\string"8000
\def\vcentcolon{\mathrel{\mathop\ordinarycolon}} \begingroup
\catcode`\:=\active \lowercase{\endgroup \let :\vcentcolon }

\newtheorem{theorem}{Theorem}[section]

\newtheorem{definition}[theorem]{Definition}

\newtheorem{claim}[theorem]{Claim}
\newtheorem{lemma}[theorem]{Lemma}

\newtheorem{fact}[theorem]{Fact}

\newcommand{\smfrac}[2]{\mbox{$\frac{#1}{#2}$}}

\newcommand{\ket}[1]{|#1\rangle}
\newcommand{\bra}[1]{\langle#1|}

\newcommand{\ra}{\rangle}

\newcommand{\ketbra}[2]{|#1\rangle\langle#2|}
\newcommand{\braket}[2]{\langle {#1} | {#2} \rangle}

\newcommand{\QMA}{{\sf{QMA}}}
\newcommand{\BPP}{{\sf{BPP}}}
\newcommand{\BQP}{{\sf{BQP}}}
\newcommand{\Pclass}{{\sf{P}}}

\newcommand{\NP}{{\sf{NP}}}
\newcommand{\comclass}{{\sf{C}}}
\newcommand{\QCMA}{{\sf{QCMA}}}

\def\final{{\mathrm{final}}}
\def\init{{\mathrm{init}}}

\newcommand{\calK}{{\cal{K}}}

\def\mns{{\mbox{-}}}

\newcommand{\kdlocal}[2]{${#2}$-state ${#1}$-local}
\newcommand{\kdLOCAL}[2]{${#2}$-STATE ${#1}$-LOCAL HAMILTONIAN}
\newcommand{\rDHam}[2]{${#1}$-dim ${#2}$-state}
\newcommand{\rDHAM}[2]{${#1}$-DIM ${#2}$-STATE HAMILTONIAN}
\newcommand{\rDqubit}[1]{${#1}$-dim}
\newcommand{\rDQUBIT}[1]{${#1}$-DIM HAMILTONIAN}

\newcommand{\mysymbol}[1]{{\mbox{\raisebox{-0.3em}{\epsfysize=1.2em\epsfbox{#1}}}}}

\newcommand{\gateflag}{\mysymbol{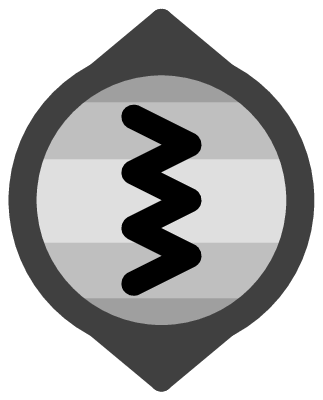}}
\newcommand{\start}{\mysymbol{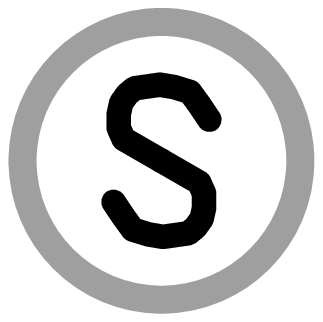}}
\newcommand{\rflag}{\mysymbol{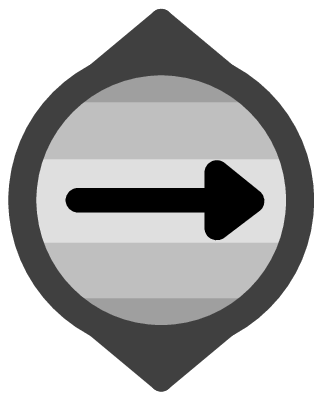}}
\newcommand{\lflag}{\mysymbol{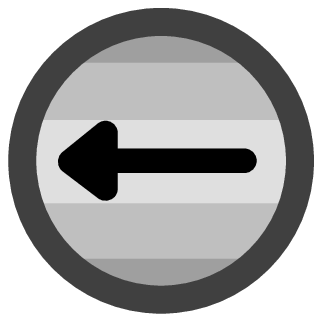}}
\newcommand{\turn}{\mysymbol{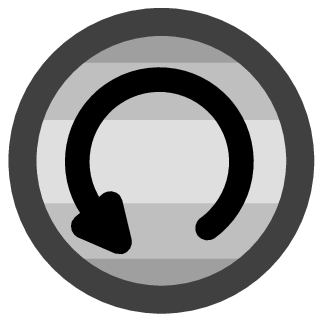}}
\newcommand{\unborn}{\mysymbol{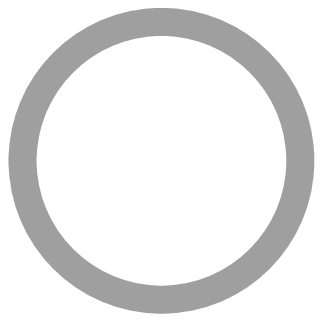}}
\newcommand{\dead}{\mysymbol{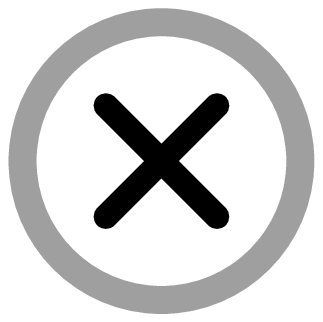}}
\newcommand{\waitl}{\mysymbol{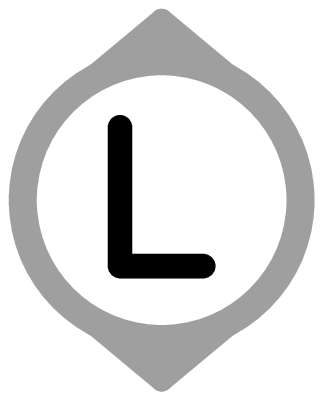}}
\newcommand{\waitr}{\mysymbol{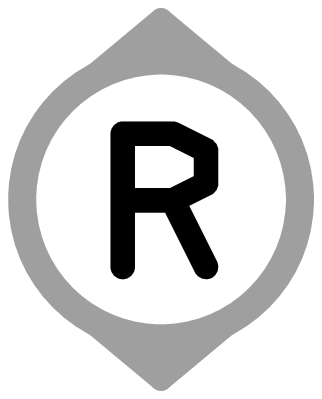}}
\newcommand{\block}{\mysymbol{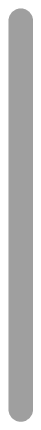}}
\newcommand{\cdotiki}{\raisebox{0.18em}{\hskip0.16em $\cdots$\hskip-0.1em }}

\newcommand{\prop}{{\mathrm{prop}}}
\newcommand{\rmpenalty}{{\mathrm{penalty}}}

\newcommand{\Hinit}{H_\init}
\newcommand{\Hfinal}{H_\final}
\newcommand{\Hprop}{H_\prop}
\newcommand{\Hpenalty}{H_\rmpenalty}

\newcommand{\Cmod}[1]{\tilde{C}_{#1}}

\newcommand{\focsv}[1]{{}}

\title{The power of quantum systems on a line}

\author{
Dorit Aharonov\thanks{E-mail: doria@cs.huji.ac.il.  Supported by
Israel Science Foundation grant number 039-7549, Binational Science
Foundation grant number 037-8404, and US Army Research Office grant
number 030-7790.} \\ School of Computer Science and Engineering,\\
Hebrew University, Jerusalem,
 Israel
\and Daniel Gottesman\thanks{E-mail:
dgottesman@perimeterinstitute.ca.  Supported by CIFAR, by the
Government of Canada through NSERC, and by the Province of Ontario
through MRI.}\\Perimeter Institute\\Waterloo, Canada
\and Sandy Irani\thanks{E-mail: irani@ics.uci.edu.  Partially
supported by NSF Grant CCR-0514082.}\\Computer Science
Department\\University of California, Irvine, USA
\and Julia Kempe\thanks{E-mail: kempe@cs.tau.ac.il. This work was
mainly done while the author was at CNRS and LRI, University of
Paris-Sud, Orsay, France. Partially supported by the European
Commission under the Integrated Project Qubit Applications (QAP)
funded by the IST directorate as Contract Number 015848, by an ANR
AlgoQP grant of the French Research Ministry, by an Alon Fellowship
of the Israeli Higher Council of Academic Research by an
Individual Research grant of the ISF, and by a European Research Council (ERC) Starting Grant.}\\School of Computer
Science,\\ Tel Aviv University, Israel }

\begin{document}

\maketitle

\begin{abstract}
We study the computational strength of quantum particles (each of
finite dimensionality) arranged on a line.  First, we prove that it
is possible to perform universal adiabatic quantum computation using
a one-dimensional quantum system (with $9$ states per particle).
This might have practical implications for experimentalists
interested in constructing an adiabatic quantum computer.  Building
on the same construction, but with some additional technical effort
and $12$ states per particle, we show that the problem of
approximating the ground state energy of a system composed of a line
of quantum particles is $\QMA$-complete; $\QMA$ is a quantum
analogue of $\NP$. This is in striking contrast to the fact that the
analogous classical problem, namely, one-dimensional MAX-$2$-SAT
with nearest neighbor constraints, is in $\Pclass$. The proof of the
$\QMA$-completeness result requires an additional idea beyond the
usual techniques in the area: Not all illegal configurations can be
ruled out by local checks, so instead we rule out such illegal
configurations because they would, in the future, evolve into a
state which can be seen locally to be illegal.  Our construction
implies (assuming the quantum Church-Turing thesis and that quantum computers cannot efficiently
solve $\QMA$-complete problems) that there are one-dimensional systems which take an
exponential time to relax to their ground states at any temperature,
making them candidates for being one-dimensional spin glasses.
\end{abstract}


\section{Introduction}

The behavior of classical or quantum spin systems frequently depends very
heavily on the number of spatial dimensions available. In
particular, there is often a striking difference between the
behavior of one-dimensional systems and of otherwise similar two- or
higher-dimensional systems. For instance, one-dimensional systems
generally do not experience phase transitions except at zero
temperature, whereas phase transitions are common in other
dimensions.  As another example, satisfiability with nearest
neighbor constraints between constant-size variables set on a two- or larger-dimensional grid is a
hard problem ($\NP$-complete, in fact), whereas for a
one-dimensional line, it can be solved in polynomial time.

We thus ask: what is the computational power
associated with a line of quantum
particles? There are a number of ways to interpret the question. For
instance, we can ask whether the evolution of such systems can be
efficiently simulated by a classical computer. We can ask whether
the system is powerful enough to create a universal quantum computer
under various scenarios for how we may control the system --- in
which case we cannot hope to simulate this behavior efficiently on a
classical computer, unless of course $\BQP=\BPP$. ($\BPP$ and $\BQP$
are the classes of problems efficiently solvable on a probabilistic
classical computer and on a quantum computer, respectively.) We can
also ask how difficult it is to calculate or approximate interesting
properties of the system, such as its ground state energy (that is,
the lowest eigenvalue of the Hamiltonian operator for the system).

For many one-dimensional quantum systems, it {\em is} indeed
possible to classically compute many properties of the system,
including in some cases the dynamical behavior of the system. The
method of the density matrix renormalization group
(DMRG)~\cite{White92,White93,DMRG-overview} has been particularly
successful here, although there is no complete characterization of
which one-dimensional systems it will work on. Indeed, DMRG provides
a good example of the difference between one- and two-dimensional
systems; there are only a few results applying DMRG techniques to
simulate special two-dimensional systems.

However, it has long been known that one-dimensional quantum systems
can also, under the right circumstances, perform universal quantum
computation.  It is straightforward to create a quantum computer
from a line of qubits with fully controllable nearest-neighbor
couplings by using SWAP gates to move the qubits around.  Even a
one-dimensional quantum cellular automaton can perform universal
quantum computation~\cite{watrous,werner:automaton}; the smallest
known construction has $12$ states per cell. While many
one-dimensional systems are relatively simple and can be simulated
classically, the general one-dimensional quantum system must thus
have complexities that are inaccessible classically.

In two interesting (and closely related) subfields of the area of
quantum computation, studied extensively over the past few years, it
has been conjectured that one-dimensional systems are not as
powerful as systems of higher dimensionality.  These cases are
adiabatic evolution, and the $\QMA$-completeness of the local
Hamiltonian problem.

%

\subsection{Results related to adiabatic computation}
In an
adiabatic quantum computer, the Hamiltonian of the system is slowly
shifted from a simple Hamiltonian, whose ground state is easy to
calculate and to create, to a more complicated Hamiltonian, whose
ground state encodes the solution to some computational problem. The
quantum adiabatic theorem guarantees that if the Hamiltonian is
changed slowly enough, the system stays close to its ground state;
the time required to safely move from the initial Hamiltonian to the
final Hamiltonian is polynomial in the minimal spectral gap between
the ground state and the first excited state over the course of the
computation.

The first adiabatic quantum algorithm was introduced by Farhi {\it
et al.}~\cite{farhiad,farhi}, though the
idea of encoding the solution to an optimization
problem in the ground state of a Hamiltonian appeared
as early as 1988 \cite{anneal, stochastic}.
The computational power of the
model was clarified in~\cite{ad1}, where it was shown that adiabatic
quantum computers could in fact perform any quantum computation.
Adiabatic computation can thus be viewed as an alternative model to
quantum circuits or quantum Turing machines.

Adiabatic quantum
computers might be more robust against certain kinds of
noise~\cite{preskill}, and the slow change of parameters used in an
adiabatic quantum computer might be more amenable to some kinds of
experimental implementation.  To this end, it is important to
understand what sorts of physical systems can be used to build an
adiabatic quantum computer, and in particular, we would like the
interactions of the physical Hamiltonian to be as simple as
possible. Several research groups have devoted effort to this
question recently. To describe their work, let us make the following
definition:
\begin{definition}[Local Hamiltonian]
Let $H$ be a Hermitian operator (interpreted as a Hamiltonian,
giving the energy of some system). We say that $H$ is an $r$-state
Hamiltonion if it acts on $r$-state particles. When $r=2$, namely, when
the particles are qubits, we often omit mention of the number of
states per particle. We say that $H$ is $k$-local
if it can be written as $H = \sum_i H_i$, where each $H_i$ acts
non-trivially on at most $k$ particles. Note that this term does
not assume anything about the physical location of the particles.
We say that $H$
is a $d$-dim Hamiltonian
if the particles are arranged on a $d$-dimensional grid and the terms
$H_i$ interact only pairs of nearest neighbor particles.
Note that the term $d$-dim implies that the
Hamiltonian is $2$-local.
\end{definition}

We will in this paper assume for simplicity that
each local term $H_i$ in the Hamiltonian is positive definite and of
norm at most 1. It is always possible to ensure this by shifting and
rescaling the Hamiltonian $H \mapsto \lambda (H+E_0)$.  Provided we
also similarly shift and rescale all other energies in the problem
(such as the spectral gap), this transformation will not at all
alter the properties of the Hamiltonian.  We will further assume, to
avoid pathologies, that there are at most polynomially many terms
$H_i$ and that the coefficients in each $H_i$ are specified to
polynomially many bits of precision.

\cite{ad1} made a first step towards a practical Hamiltonian, by
showing that a \rDHam{2}{6} Hamiltonian suffices for universal
adiabatic quantum computation.  Kempe, Kitaev and Regev
\cite{KempeKitaevRegev06}, using perturbation-theory gadgets,
improved the results and showed that qubits can be used instead of
six-dimensional particles, namely, adiabatic evolution with
$2$-local Hamiltonians is quantum universal. Their interactions,
however, were not restricted to a two-dimensional grid. Terhal and
Oliveira~\cite{Oliveira:05a} combined the two results, using other
gadgets, to show that a similar result holds with \rDHam{2}{2}
Hamiltonians.

Naturally, the next question was whether a one-dimensional adiabatic
quantum computer could be universal.  Since one-dimensional systems are fairly
simple, it was conjectured that the answer was ``no.''
However, in this paper, we prove
\begin{theorem}\label{thm:adiabatic}
Adiabatic computation with \rDHam{1}{9} Hamiltonians is universal
for quantum computation.
\end{theorem}

As mentioned, this could be important experimentally, as
one-dimensional systems are easier to build than two-dimensional
ones, and adiabatic systems might be more robust than standard ones.
Any simplification to the daunting task of building a quantum
computer is potentially useful.  However, a more systematic way of
dealing with errors will be needed before it is possible to build a
large adiabatic quantum computer. The results of \cite{preskill}
only imply adiabatic quantum computation is robust against certain
types of control errors, and it remains an interesting open question
to show that adiabatic quantum computation can be made
fault-tolerant in the face of general small errors.  An important step in
this direction is \cite{shor:adia}, which introduces some quantum
error-correcting codes in the adiabatic model.

%
It is also worth noting that the adiabatic construction given in the proof
of Theorem~\ref{thm:adiabatic} actually has a degenerate ground state,
which makes it less robust against noise than a non-degenerate adiabatic computer, although still resistant
to timing errors, for instance.  Using a \rDHam{1}{13} Hamiltonian similar to the $\QMA$-complete
\rDHam{1}{12} Hamiltonian family described below, it is possible to break the degeneracy. This construction
is outlined in Section \ref{sec:nondegenerate}.

It is natural to ask whether our result holds even when we restrict
ourselves to translationally invariant Hamiltonians.
We sketch one way to do this in
Section~\ref{sec:protocol}.
However, it does not appear to be possible to make the translational
invariant adiabatic construction nondegenerate.

%

\subsection{Results related to $\QMA$-completeness}
Theorem~\ref{thm:adiabatic} means that efficient simulation of
general one-dimensional adiabatic quantum systems is probably
impossible. One might expect that calculating, or at least
approximating, some specific property of a system, such as its
ground state energy, would be more straightforward, as this does not
require complete knowledge of the system.  This intuition is
misleading, and in fact it can be {\em harder} to find the ground
state energy than to simulate the dynamics of a system. More
concretely, it has long been known that it is $\NP$-hard to find the
ground state energy of some classical spin systems, such as a
three-dimensional Ising model. Kitaev~\cite{Kitaev:book} extended
these results to the quantum realm, by defining the quantum analogue
of $\NP$, called $\QMA$ (for ``quantum Merlin-Arthur''). $\QMA$ is
thus, roughly speaking, the class of problems that can be
efficiently checked on a quantum computer, provided we are given a ``witness'' quantum state related to the answer to the problem. $\QMA$ is believed to be
strictly larger than both $\BQP$ and $\NP$. That is, a problem which
is $\QMA$-complete can most likely not be solved efficiently with a
quantum computer, and the answer cannot even be checked efficiently
on a classical computer.

Kitaev proved that the problem of approximating the ground state
energy of a quantum system with a $5$-local Hamiltonian is complete
for this class. (The ability to solve a $\QMA$-complete problem
implies the ability to solve any problem in $\QMA$.) The exact
definition of the local Hamiltonian problem is this:

\begin{definition}[Local Hamiltonian problem]
\label{def:localHam} Let $H = \sum_i H_i$ be an \kdlocal{k}{r}
Hamiltonian. Then $(H, E, \Delta)$ is in \kdLOCAL{k}{r} if the
lowest eigenvalue $E_0$ of $H$ is less than or equal to $E$.  The
system must satisfy the promises that $\Delta = \Omega(1/{\rm
poly}(n))$ and that either $E_0 \leq E$ or $E_0 \geq E + \Delta$. We
can define a language \rDHAM{d}{r} similarly.
\end{definition}

Kitaev's proof uses what is called the circuit-to-Hamiltonian
construction of Feynman \cite{feynman}, which we sketch in Subsection~\ref{outline}.
The main result of \cite{ad1}, namely the universality of
adiabatic computation,
was based on the observation that the circuit-to-Hamiltonian construction
is useful also in the adiabatic context, where it is used
to design the final Hamiltonian of the adiabatic evolution.
It turns out that the same
connection between adiabatic universality proofs and $\QMA$-completeness
proofs can also be made for many of the follow-up papers on the two topics.
It is often the case that
whenever adiabatic universality can be proven for some class
of Hamiltonians, then the local Hamiltonian problem with the
same class (roughly) can be shown to be $\QMA$-complete --- and vice versa.

Note, however, that there is no formal implication from either of those
problems to the other. On one hand, proving $\QMA$-completeness is in general
substantially harder than achieving universal adiabatic quantum
computation: In an adiabatic quantum computer, we can choose the
initial state of the adiabatic computation to be any easily-created
state which will help us solve the problem, so we can choose to work
on any convenient subspace which is invariant under the Hamiltonian.
For $\QMA$, the states we work with are chosen adversarially from
the full Hilbert space, and we must be able to check ourselves,
using only local Hamiltonian terms, that they are of the correct
form. On the other hand, proving adiabatic universality involves
analyzing the spectral gap of all of the Hamiltonians $H(t)$ over the
duration of the computation, whereas $\QMA$-completeness proofs are only concerned with one
Hamiltonian.

Both~\cite{KempeKitaevRegev06} and~\cite{Oliveira:05a} in fact
prove $\QMA$-completeness; universal adiabatic quantum computation
follows with a little extra work. Thus, the \rDQUBIT{2} problem is
$\QMA$-complete.

However, the \rDHAM{1}{r} problem remained open for all (constant)
$r$. It was suspected that the problem was not $\QMA$-complete, and
in fact might be in $\BQP$ or even $\BPP$.  For one thing, ground
state properties of one-dimensional quantum systems are generally
considered particularly easy.  For instance, Osborne has recently
proven~\cite{osborne2} that there are efficient classical
descriptions for a class of one-dimensional quantum systems.  DMRG
techniques have been employed extensively to calculate ground state
energies and other properties of a variety of one-dimensional
quantum systems. Furthermore, the classical analogue of \rDHAM{1}{r}
is easy: Take $r$-state variables arranged on a line with
constraints restricting the values of neighboring pairs of
variables.  If we assign a constant energy penalty for violating a
constraint, the lowest energy state satisfies as many constraints as
possible.  This problem, a one-dimensional restriction of
MAX-$2$-SAT with $r$-state variables, is in fact in $\Pclass$; it
can be solved with a recursive divide-and-conquer algorithm or by
dynamic programming.

For instance, we can divide the line in half, and run through all
possible assignments of the two variables $x_i$ and $x_{i+1}$
adjacent to the division.  For each pair of values for $x_i$ and
$x_{i+1}$, we can calculate the maximal number of satisfiable
constraints by solving the sub-problems for the right and left
halves of the line.  Then we can compare across all $r^2$
assignments for the one that gives the largest number of satisfied
constraints.  Thus, to solve the problem for $n$ variables, it is
sufficient to solve $2r^2$ similar problems for $n/2$ variables.  By
repeatedly dividing in half, we can thus solve the problem in
$O((2r^2)^{\log n}) = O(n^{\log (2r^2)})$ steps.  MAX-$k$-SAT
in one dimension is also in $\Pclass$, and can in fact be reduced to
MAX-$2$-SAT with large enough variables.

Despite the intuition that one-dimensional systems should not be too
difficult, we prove:

\begin{theorem}\label{thm:qma}
\rDHAM{1}{12}\ is $\QMA$-complete.
\end{theorem}

The theorem implies a striking qualitative difference between the quantum
and the classical one-dimensional versions of the same problem ---
one is easy, whereas the other is complete for a class which seems
to be strictly larger than $\NP$. This might seem surprising, but in
retrospect, we can provide an intuitive explanation. The reason is
that the $k$-local Hamiltonian essentially allows us to encode an
extra dimension, namely, time, by making the ground state a
superposition of states corresponding to different times. In other
words, the result implies that the correct analogue of
one-dimensional local Hamiltonian is {\it two}-dimensional
MAX-$k$-SAT, which is of course $\NP$-complete.  Indeed,
there are many cases in physics where one-dimensional quantum
systems turn out to be most closely analogous to two-dimensional
classical systems. One such example is given by Suzuki who showed that
the the $d$-dimensional  Ising model with a transverse field is
equivalent to the $(d+1)$-dimensional Ising model at finite temperatures
\cite{suzuki}, although the extra dimension there is  quite different
from the dimension of time addressed in this paper.

\subsection{Implications of our results}
One consequence of Theorem~\ref{thm:qma} is that there exist one-dimensional systems which take exponentially long to relax to the ground state, with any reasonable environment or cooling strategy.  To see this, we invoke a quantum version of the modern Church-Turing thesis,%
\footnote{The modern Church-Turing thesis is discussed, for instance, in~\cite{MCTthesis}.}
which would state that any reasonable physical system can be efficiently simulated on a quantum computer.  As it invokes a notion of a ``reasonable'' physical system, this is of course not a provable statement, but one can be convinced of it by looking at examples.  The details of doing so for any particular system may be complicated, but no counter-example is known, and barring undiscovered laws of physics modifying quantum mechanics, the quantum Church-Turing thesis is believed to be true.

In this particular case, we could thus use a quantum computer to simulate the system plus the environment.  If the system reliably relaxes to the ground state in a time $T$, then by causality, we need only simulate an environment of size $O(T^3)$ (for a three-dimensional environment); the quantum simulation can be performed using standard techniques and runs in polynomial time.  We can then use the simulation to solve \rDHAM{1}{12}\ problems. Since the latter problem is difficult, that implies the simulation must also take a long time to reach the ground state.
This observation assumes, of course, that $\QMA$ is exponentially hard to solve by a quantum computer.

Similarly, we can argue that the class of systems presented in this
paper will take an exponentially long time to relax to the
thermal equilibrium state at very low temperatures (again,
assuming $\QMA$ is hard for a quantum computer).  To see this,
suppose that there is a state of the system with energy less than
$E$ (the bound on the ground state energy from
Definition~\ref{def:localHam}), and note that at a temperature less
than $O(\Delta/n)$ (with $\Delta$ polynomially small in the system
size $n$), the Gibbs distribution gives a constant probability of the
system occupying the low-lying state, even if there are exponentially
many eigenstates with energy just above $E + \Delta$.  Of course,
if there is no state with energy less than $E+\Delta$, then the system
will never have energy less than $E+\Delta$, no matter what the temperature.
Therefore, if the system were able to reach its thermal equilibrium state,
by observing it, we would be sampling from the Gibbs distribution, and
would, with good probability, be able to answer the $\QMA$-complete
\rDHAM{1}{12}\ problem.

Exponentially long relaxation times are a characteristic of a {\it spin
glass}~\cite{spinglassreview}, making this system a candidate to be
a one-dimensional spin glass. Spin glasses are not completely
understood, nor is there a precise definition of the term, but it
has been unclear whether a true one-dimensional spin glass can
exist.  There are a number of properties relevant to being a spin
glass, long relaxation times among them, and classical
one-dimensional systems generally lack that property and most
other typical spin glass properties.  On the other hand, some other
one-dimensional quantum systems are known~\cite{fisher} that exhibit
long relaxation times and some, but not all, other spin-glass-like properties.
Our result may help shed further light on the situation.  Unfortunately,
we are unable to present a specific Hamiltonian with long relaxation
times. As usual in complexity, it is difficult to identify specific hard
instances of computationally hard problems. Similarly, we cannot say very
much about other physically interesting properties of the system,
such as the nature or existence of a liquid/glass phase transition.

\medskip
This paper is a merger of \cite{irani} and version $1$ of
\cite{AGK}, two similar papers originally written independently.
The extended abstract of the  merged version appeared in
\cite{focsVersion}.  Following the initial version of our paper, various
groups have improved and extended our result.  For instance,
\cite{NWtranslation} gives a $20$-state translation-invariant modification of our
construction (improving on a $56$-state construction by \cite{JWZtranslation})
that can be used for universal 1-dimensional adiabatic computation.
Kay~\cite{Kaytranslation} gives a $\QMA$-complete \rDHam{1}{r} Hamiltonian
with all two-particle terms identical.  Nagaj~\cite{Nagajthesis} has improved
our $\QMA$ construction to use $11$-state particles.

\subsection{Notation}

A {\em language} is computer science terminology for a computational
problem; it can be defined as a set of bit strings. Frequently we
consider the language to be a proper subset of all possible bit
strings, but here we will be more concerned with languages which
satisfy a {\em promise}, meaning we consider the language as a
subset of a smaller set of bit strings which satisfy some
constraint. We will refer to the set of strings which satisfy the
constraint as the {\em instances} of the problem. The instances then
can be  partitioned into two sets $L_{\rm yes}$, the set of instances which are
in the language,
and $L_{\rm no}$, the set of instances which are not in the language.
The {\em decision problem} associated with the language
is to decide if a given bit string is in the language or not.  The instances
represent different possible inputs about which we are attempting to
reach a yes/no decision. For example,
satisfiability (or SAT) is the decision problem for the language
composed of Boolean formulas which have satisfying assignments of
the variables; each instance encodes a possible formula, some of
which can be satisfied and some of which cannot.  A {\em
complexity class} is a set of languages.

The precise definition of the complexity class $\QMA$ is as follows:

\begin{definition}
A language $L$ is in $\QMA$ iff for each instance $x$ there exists
a uniform polynomial-size quantum circuit%
\footnote{A {\em uniform} circuit is one whose description can be
generated in polynomial time by a Turing machine.}
$C_x$ such that (a) if $x \in L_{\rm yes}$, $\exists \ket{\psi}$ (the
``witness'', a polynomial-size quantum state) such that $C_x
\ket{\psi}$ accepts with probability at least $2/3$, and (b) if $x
\in L_{\rm no}$, then $\forall \ket{\psi}$, $C_x \ket{\psi}$ accepts with
probability at most $1/3$.  We only consider strings which are instances
of $L$.
\end{definition}

We say we {\em reduce} language $L$ to language $M$ if we have a
function $f$ converting bit strings $x$ to bit strings $f(x)$ such
that if $x \in L_{\rm yes}$ then $f(x) \in M_{\rm yes}$ and if $x \in L_{\rm no}$,
then $f(x) \in M_{\rm no}$.  We do not put any requirement on the behavior
of $f$ when the input $x$ does not satisfy the promise.  $f$ must be computable
classically in polynomial time.  The point is that if we can solve
the decision problem for $M$ then we can automatically also solve
the decision problem for $L$ by first converting $x$ to $f(x)$ and
then checking if $f(x)$ is in $M$.  There are other notions of
reduction, but this will be sufficient for our purposes.  A language
$L$ is {\em complete} for a complexity class $\comclass$ if any
language in $\comclass$ can be reduced to $L$.  Thus, solving a
$\comclass$-complete language $L$ implies the ability to solve any
problem in $\comclass$ with polynomial time overhead; $L$ is thus
one of the most computationally difficult problems in $\comclass$.

We will be interested in reducing an arbitrary $\QMA$ language $L$
to variants of the language \kdLOCAL{k}{r}.  We will use $x$ to
refer to an instance of the original problem $L$, and $C_x$ to refer
to the checking circuit associated with the instance $x$.  That is,
we assume that information about the instance is encoded in the
structure of the checking circuit.  The circuit $C_x$ acts on $n$
qubits. Some of the qubits will be ancilla qubits which must be
initialized to the state $\ket{0}$, whereas others are used for the
potential witness for $x \in L$.

In the context of adiabatic quantum computation, we will also use
$C_x$ to refer to a more general quantum circuit, unrelated to a
$\QMA$ problem. In this case, we consider $C_x$ to be in a form
where all input qubits are $\ket{0}$. When the adiabatic quantum
computation is intended to compute some classical function (such as
factoring), $x$ is the classical input to the function, which we
hardwire into the structure of the circuit itself.

\subsection{Outline of the approach}
\label{outline}

To explain the main idea behind the proofs of our two main theorems,
we recall Kitaev's proof that $5$-LOCAL HAMILTONIAN is \QMA-complete.

The fact that \kdLOCAL{k}{r} is in $\QMA$ is not difficult.
The witness is the ground state of the Hamiltonian (or indeed any
state with energy less than $E$), and a standard phase estimation
technique can be used to measure its energy.  The accuracy of the
measurement is not perfect, of course, which is why we need the
promise that the energy is either below $E$ or above $E + \Delta$.

To prove completeness, we need to reduce an arbitrary $\QMA$
language $L$ to $5$-LOCAL HAMILTONIAN.  The translation to local
Hamiltonians is done by creating a Hamiltonian whose ground state is
of the form $\sum_t \ket{\phi_t} \ket{t}$ (ignoring normalization),
where the first register $\ket{\phi_t}$ is the state of the checking
circuit $C_x$ at time $t$ and the second register acts as a clock.
We call this state the
{\em history state} of the circuit.%
\footnote{We call a state a history state if it has the above form
for {\em any} input $\ket{\phi_0}$ to $C_x$, not just the correct
input.}
The main term in Kitaev's Hamiltonian sets up constraints ensuring
that $\ket{\phi_{t+1}} = U_t \ket{\phi_{t}}$, where $U_t$ is the
$(t+1)$-st gate in $C_x$. This term, denoted $\Hprop$, ensures that
the ground state is indeed a history state, reflecting a correct
propagation in time according to $C_x$. The clock is used to
associate the correct constraint with each branch of the
superposition; any state which does not have the correct time
evolution for the circuit will violate one or more constraints and
will thus have a higher energy.

Kitaev's Hamiltonian includes more terms, which guarantee that the
input to $C_x$ is correct, that the final state of the circuit is
accepted, and that the state of the clock is a valid clock state and
not some arbitrary state.  In the context of adiabatic evolution,
these additional terms are not needed, since we have control over
the initial state, but they are needed to prove $\QMA$-completeness,
since we must be able to check (i) that the witness being tested for
the $\QMA$-complete problem has the correct structure, corresponding
to a valid history state for a possible input to $C_x$, and (ii)
that $C_x$ accepts on that input.

In order to prove the universality of adiabatic quantum computation,
\cite{ad1} let Kitaev's Hamiltonian $\Hprop$ be the final
Hamiltonian in the adiabatic evolution, and set the initial
Hamiltonian to be diagonal in the standard basis, forcing the
initial ground state to be the correct input state
$\ket{\phi_0}\ket{0}$ of the circuit. At any point during the
adiabatic quantum computation, the state of the system is then in a
subspace ${\cal I}$ spanned by the states $\ket{\phi_t}\ket{t}$, and
the spectral gap of any convex combination of the initial and final
Hamiltonians restricted to this subspace is at most polynomially
small in the size of the original circuit. Therefore, the adiabatic
quantum computation can be performed with at most polynomial
overhead over the original circuit, to reach a state which is very
close to the history state, from which the result of the original
circuit can be measured with reasonable probability.

Our idea for proving universality of adiabatic evolution and
$\QMA$-completeness in one dimension is similar. One would consider
a quantum circuit $C_x$, and design a \rDHam{1}{r} Hamiltonian which
will verify correct propagation according to this circuit, and then
use this Hamiltonian as the final Hamiltonian for the adiabatic
evolution or as the instance of the $\QMA$-complete problem
corresponding to the instance $x$ of the original problem.

This idea, however, is not easy to realize when our Hamiltonian is
restricted to work on a one- or two-dimensional lattice, since we
cannot directly access a separate clock register, as only the
subsystems nearest to the clock would be able to take advantage of
it in order to check correct propagation in time.  Instead,
following the strategy of~\cite{ad1}, we must first modify the
original circuit $C_x$ into a new circuit $\Cmod{x}$.  This allows
us to distribute the clock, making the time implicit in the global
structure of the state. When time is encoded somehow in the
configuration, then it can be ensured that the transition rules
allow only propagation to the correct next step. The construction of
\cite{ad1} used the following arrangement: the qubits of the
original circuit were put initially in the left-most column in the
two-dimensional grid, and one set of gates was performed on them. To
advance time and perform another set of gates, the qubits were moved
over one column, leaving used-up ``dead'' states behind. Doing all
this required going up to $6$ states per particle instead of $2$,
with the qubits of the original circuit encoded in various
two-dimensional subspaces.  The time could therefore be read off by
looking at the arrangement and location of the two-dimensional
subspaces containing the data.

This construction relies heavily on the ability to copy qubits to
the next column in order to move to the next block of gates in the
computation, so a new strategy is needed in one dimension. For the
modified circuit $\Cmod{x}$, we instead place the qubits in a block
of $n$ adjacent particles. We do one set of gates, and then move all
of the qubits over $n$ places to advance time in the original
circuit $C_x$.  It is considerably more complicated to move qubits
$n$ places than to move them over one column in a two-dimensional
arrangement, and we thus need extra states.  The adiabatic
construction can be done using $9$ states per particle, which we
increase to $12$ for $\QMA$-completeness.

A straightforward application of existing techniques can then
complete the proof of one-dimensional universal adiabatic quantum
computation, but there is an additional wrinkle for proving
$\QMA$-completeness. In particular, previous results used local
constraints to ensure that the state of the system had a valid
structure; for instance, in \cite{ad1}, terms in the Hamiltonian
check that there are not two qubit states in adjacent columns.
However, using only local constraints, there is no way to check that
there are exactly $n$ qubit data states in an unknown location in a
one-dimensional system --- there are only a constant number of local
rules available, which are therefore unable to count to an
arbitrarily large $n$. We instead resort to another approach, which
might be useful elsewhere.  While our modified circuit has invalid
configurations (containing, for instance, too many qubit states)
which cannot be locally checked, we ensure that, under the
transition rules of the system, any invalid configurations will
evolve in polynomial time into a configuration which can be detected
as illegal by local rules. Thus, for every state which is not a
valid history state, either the propagation is wrong, which implies
an energy penalty due to the propagation Hamiltonian, or the state
evolves to an illegal configuration which is locally detectable,
which implies an energy penalty due to the local check of illegal
configurations. We call this result the {\em clairvoyance lemma}.

\medskip
The structure of the paper is as follows: in Section \ref{sec:protocol} we describe how to map circuits to
one-dimensional arrangements. Section \ref{sec:adiabatic} shows our result on the universality of adiabatic
computation on the line, and Section \ref{sec:QMA} gives the result on $\QMA$-completeness. We conclude in
Section \ref{sec:conclude}.

\section{The basic construction}
\label{sec:protocol}

In this section, we will describe our construction which maps a
quantum circuit $C_x$ to a modified circuit $\Cmod{x}$.  We first
present the larger $12$-state $\QMA$ construction, and explain how
to modify it to get the $9$-state adiabatic construction at the end
of this section.  There are a number of properties which $\Cmod{x}$
must satisfy. It should perform the same computation as $C_x$, of
course. In addition, each gate in $\Cmod{x}$ must interact only
nearest-neighbor particles in a line; on the other hand, we allow
those particles to be $12$-dimensional. The gate performed at any
given time cannot depend explicitly on the time (the number of gates
already performed), but can depend on location. We will define the
gates in terms of a variety of possible transition rules, and to
remove any ambiguity, we will ensure that for a legal state of the
system, only one transition rule will apply at any given time.  For
$\QMA$-completeness, we need some additional properties ensuring
that enough of the constraints are locally checkable.

The problem of moving to one dimension is somewhat similar to a
quantum cellular automaton in that the transition rules need to
depend only on the local environment and not on some external time
coordinate, but differs from a cellular automaton in a number of
ways.  A cellular automaton acts on all locations simultaneously and
in the same way, whereas our transition rules are required to only
cause one pair of particles to change per time step (provided the
 system is in a state of the correct form), and the
transition rules differ slightly from location to location.

A better analogy is to a single-tape Turing machine.  We will have a
single ``active site'' in our computation, analogous to the head of
a Turing machine, which moves around manipulating the computational
qubits, changing states as it does so in order to perform different
kinds of actions. The transition rules of a Turing machine are
independent of the location. In the construction below, we have a
few different kinds of locations, for instance for the different
sorts of quantum gates used in $C_x$.  Each kind of location has a
different set of transition rules that apply (although many of the
rules are the same for all types of location), and the position of a
site determines which set of rules applies to that site.  For the
adiabatic construction, we could instead proceed with
translational-invariant rules as with a Turing machine by
initializing each site with an additional marker indicating what
type of location it is supposed to be.  This would require a
substantial increase in the number of states per site (perhaps to
$100$ or so), and we do not know how to use translational invariant
nearest-neighbor rules for the $\QMA$ result. Instead we use
position-dependent rules.  This enables us to use essentially the
same construction for both the adiabatic result and the
QMA-completeness result, and reduces the required number of states
in the adiabatic case.

In the beginning we put the quantum circuit we wish to simulate into
a canonical form. Let $n$ be the number of qubits in the circuit,
labeled $1 \ldots n$ from left to right in a line. We assume that
the circuit is initialized to the all $\ket{0}$ state and consists
of $R$ rounds.  Each round is composed of $n-1$ nearest-neighbor
gates (some or all of which may be the identity gate); the first
gate in a round acts on qubits $1$ and $2$, the second on qubits $2$
and $3$, and so on. Any quantum circuit can be put into this form
with the number of rounds at most proportional to the number of
gates. We will assume that the first round and the last round
contain only identity gates, as we will need those two rounds for
additional checking in the $\QMA$-completeness result.

In the 1D arrangement, there will be a total of $nR$ $12$-state particles, arranged on a line in $R$ blocks
of $n$ qubits, each block corresponding to one round in the circuit. Roughly speaking, we imagine that in the
beginning the qubits are all in the first block. The first round of quantum gates is performed one by one on
these qubits, after which they are moved to the next block, where the second round of gates is performed, and
so on.  Once the qubits reach the last block and undergo the last round of gates, their state will be the
final state of the circuit.

The main difficulty we will have to overcome is our inability to
count.  Since there are only a constant number of states per site,
there is no way to directly keep track of how far we have moved. Our
solution is to move the full set of qubits over only one space at a
time. We keep moving over until the qubits reach the next block.  We
can tell that we have reached a new block, and are ready to perform
a new set of gates, by making the transition rules different at the
boundary between two blocks.  When the qubits are correctly aligned
with a block, a qubit state and a non-qubit state will be adjacent
across a block boundary, whereas while the qubits are moving,
adjacent pairs of particles which cross a block boundary will either
have two qubit states or two non-qubit states.

We denote a block boundary by $\cdot$\block$\cdot $. The $12$ states
in each site consist of $2$-state subsystems (different versions of
a qubit holding data), represented by elongated shapes (e.g.,
\waitr), and $1$-state subspaces, represented by round shapes (e.g.,
\unborn). Two of the $2$-state systems and two of the $1$-state site
types will be ``flags'' or ``active'' sites, which will be
represented by dark shapes and can be thought of as pointers on the
line that carry out the computation. Light-colored shapes represent
a site that is inactive, waiting for the active site to come nearby.
There will only be one active site in any (valid) configuration.  We
have the following types of states:
\medskip

\begin{tabular}{|l|l|}
  \hline
 Inactive states  & Flags (active states) \\ \hline
  \waitr: Qubits to right of active site & \gateflag: Gate marker (moves right) \\
  \waitl: Qubits to left of active site & \rflag: Right-moving flag\\
  \unborn: Unborn (to right of all qubits) &  \lflag: Left-moving flag, moves qubits right one space \\
  \dead: Dead (to left of all qubits) & \turn: Turning flag \\
  \hline
\end{tabular}

\medskip
\noindent The \gateflag\ and \rflag\ active sites are qubit states.
In the analogy to a Turing machine, they can be thought of as the
head sitting on top of a qubit on the tape.  The \lflag\ and \turn\
flags are one-dimensional subspaces, which sit between or next to
the particles which are in qubit states.

\begin{definition}
\label{def:legal} We use the term {\em configuration} to refer to an
arrangement of the above types of states without regard to the value
of the data stored in the qubit subsystems. {\em Valid} (or {\em
legal}) configurations of the chain have the following structure:
\begin{list}{}{}
\item \dead \cdotiki \dead ({\bf qubits})\unborn \cdotiki \unborn,
\end{list}
\noindent where the \dead \cdotiki \dead\ string  or \unborn \cdotiki \unborn\ string might not appear if the
qubits are at one end of the computer. The ({\bf qubits}) string consists of either $n$ sites (for the first
two possibilities) or $n+1$ sites (for the last three choices) and is of one of the following forms:

\begin{tabular}{l@{}c@{}r@{}l}
 \waitl \cdotiki & \waitl \gateflag \waitr & \cdotiki \waitr &\\

\waitl \cdotiki & \waitl \rflag \waitr & \cdotiki \waitr &\\

 \waitl \cdotiki & \waitl \lflag \waitr & \cdotiki \waitr & \waitr\\

 \turn \waitr & \cdotiki & \waitr & \waitr\\

 \waitl & \cdotiki &  \waitl & \turn
\end{tabular}

\noindent In the first three cases, either the \waitl\ string or the \waitr\ string might be absent when the
active site is at one end of the qubits. (The \turn\ flag is only needed at the left or right end of a string
of qubits.)  When the flag is \gateflag, the $n$ qubit sites are lined up inside a single block; for all
other values of the flag, the ({\bf qubits}) string crosses a block boundary.
\end{definition}

The initial configuration is \block \gateflag \waitr \cdotiki \waitr
\block \unborn \cdotiki \unborn \block \cdotiki \block \unborn
\cdotiki \unborn \block\  where the \gateflag\ and the \waitr\
qubits are in the state corresponding to the input of the original
circuit $C_x$. Following the transition rules below, the \gateflag\
sweeps to the right, performing gates as it goes. When it reaches
the end of the qubit states at the border of the next block, it
becomes a \turn, which in turn creates a \lflag\ flag which sweeps
left, moving each qubit one space to the right in the process. When
the \lflag\ flag reaches the left end, it stops by turning into a
\turn\ which then creates a \rflag\ flag, which moves right through
the qubits without disturbing them. The \rflag\ flag hits the right
end of the set of qubits and becomes a \turn, which begins a \lflag\
flag moving left again. The \lflag-\turn-\rflag-\turn\ cycle
continues until the qubits have all been moved past the next block
boundary. Then we get a \turn \block \waitr\ arrangement, which
spawns a new \gateflag, beginning the gate cycle again. The
evolution stops when the qubits reach the last block boundary and
the gate flag reaches the end of the line, i.e., the final
configuration is \block \dead \cdotiki \dead \block \cdotiki \block
\dead \cdotiki \dead \block \waitl \cdotiki \waitl \gateflag \block.

We have the following transition rules.  Something of the form XY is never across a block boundary, whereas
X\block Y is.  \block Y and X\block\ represent the left and right end of the chain, respectively.

\begin{enumerate}
\item (Gate rule) \gateflag \waitr\ $\rightarrow$ \waitl \gateflag\ (performing the appropriate gate between the two encoded
qubits)

\item (Turning rules right side) \gateflag \block \unborn\ $\rightarrow$ \waitl \block \turn, \turn \unborn\
$\rightarrow$ \lflag \unborn, \turn \block \unborn\ $\rightarrow$ \lflag \block \unborn, \turn \block\
$\rightarrow$ \lflag \block

\item (Sweeping left rules)  \waitl \lflag\ $\rightarrow$ \lflag \waitr, \waitl \block \lflag\ $\rightarrow$
\lflag \block \waitr

\item (Turning rules left side) \dead \lflag\ $\rightarrow$ \dead \turn,  \dead \block \lflag\ $\rightarrow$
\dead \block \turn, \block \lflag\ $\rightarrow$ \block \turn, \turn \waitr\ $\rightarrow$ \dead \rflag

\item (Sweeping right rules) \rflag \waitr\ $\rightarrow$ \waitl \rflag, \rflag \block \waitr\ $\rightarrow$
\waitl \block \rflag, \rflag \unborn\ $\rightarrow$ \waitl \turn

\item (Starting new round rule) \turn \block \waitr\ $\rightarrow$ \dead \block \gateflag
\end{enumerate}

\begin{figure}[ht]
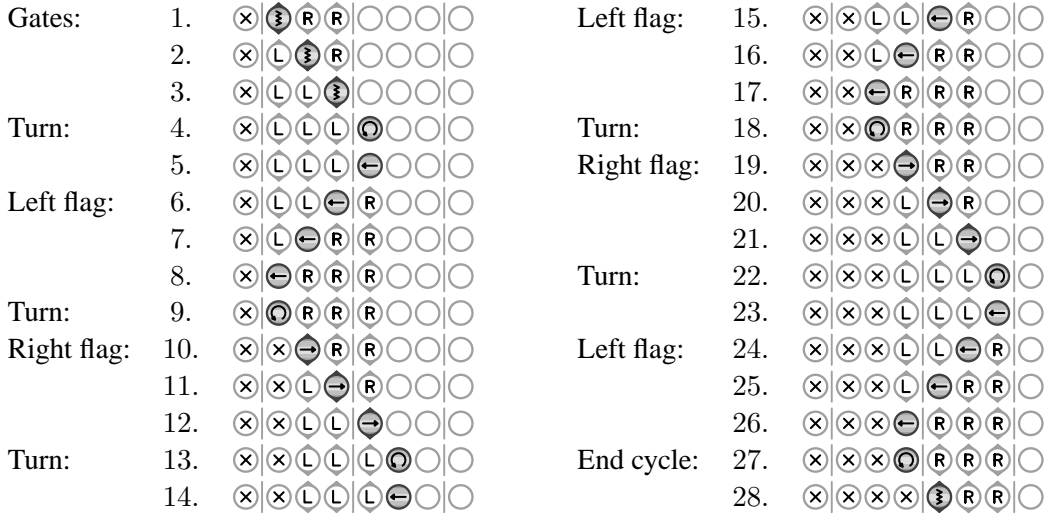

\begin{tabular}{lcr@{}c@{}c@{}c@{}lcc}
Gates: &$1.$ & \dead \block &  \gateflag \waitr \waitr & \block & \unborn \unborn \unborn & \block \unborn &
\\&$2.$ & \dead \block & \waitl \gateflag \waitr & \block & \unborn \unborn \unborn & \block\unborn &
\\&$3.$ & \dead \block & \waitl \waitl \gateflag & \block & \unborn
\unborn \unborn & \block \unborn \\
Turn: & $4.$ & \dead \block & \waitl \waitl \waitl & \block & \turn \unborn \unborn & \block\unborn &
\\&$5.$ & \dead \block & \waitl \waitl \waitl & \block & \lflag \unborn \unborn & \block \unborn &  \\
Left flag: & $6.$ & \dead \block & \waitl \waitl \lflag & \block & \waitr \unborn \unborn & \block\unborn &
\\&$7.$ & \dead \block & \waitl \lflag \waitr & \block & \waitr \unborn \unborn & \block \unborn &
\\&$8.$ & \dead \block & \lflag \waitr \waitr & \block & \waitr \unborn \unborn & \block \unborn \\
Turn: & $9.$ & \dead \block & \turn \waitr \waitr & \block & \waitr \unborn \unborn & \block \unborn & \\
Right flag: & $10.$ & \dead \block & \dead \rflag \waitr & \block & \waitr \unborn \unborn & \block\unborn &
\\&$11.$ & \dead \block & \dead \waitl \rflag & \block & \waitr \unborn \unborn & \block\unborn &
\\&$12.$ & \dead \block & \dead \waitl \waitl & \block & \rflag
\unborn \unborn & \block \unborn \\
Turn: & $13.$ & \dead \block & \dead \waitl \waitl & \block & \waitl \turn \unborn & \block\unborn &
\\&$14.$ & \dead \block & \dead \waitl \waitl & \block & \waitl \lflag \unborn & \block \unborn
&&
\end{tabular}
\begin{tabular}{lcr@{}c@{}c@{}c@{}lccccc}
Left flag: & $15.$ & \dead \block & \dead \waitl \waitl & \block & \lflag \waitr \unborn & \block\unborn &
\\&$16.$ & \dead \block & \dead \waitl \lflag & \block & \waitr \waitr \unborn & \block\unborn &
\\&$17.$ & \dead \block & \dead \lflag \waitr & \block & \waitr
\waitr \unborn & \block \unborn \\
Turn: & $18.$ & \dead \block & \dead \turn \waitr & \block &  \waitr
\waitr \unborn & \block \unborn & & & & \\
Right flag: & $19.$ & \dead \block & \dead \dead \rflag & \block & \waitr \waitr \unborn & \block\unborn &
\\&$20.$ & \dead \block & \dead \dead \waitl & \block & \rflag \waitr \unborn & \block\unborn &
\\&$21.$ & \dead \block & \dead \dead \waitl & \block & \waitl \rflag
\unborn & \block \unborn \\
Turn: & $22.$ & \dead \block & \dead \dead \waitl & \block & \waitl \waitl \turn & \block\unborn &
\\&$23.$ & \dead \block & \dead \dead \waitl & \block & \waitl \waitl \lflag & \block \unborn & & \\
Left flag: & $24.$ & \dead \block & \dead \dead \waitl & \block & \waitl \lflag \waitr & \block\unborn &
\\&$25.$ & \dead \block & \dead \dead \waitl & \block & \lflag \waitr \waitr & \block\unborn &
\\&$26.$ & \dead \block & \dead \dead \lflag & \block & \waitr \waitr
\waitr & \block \unborn \\
End cycle: & $27.$ & \dead \block & \dead \dead \turn & \block & \waitr \waitr \waitr & \block\unborn
&\\&$28.$ & \dead \block & \dead \dead \dead & \block & \gateflag \waitr \waitr & \block\unborn
\end{tabular}
\caption{A full cycle, with $n=3$. Time goes down the first and then the second column.} \label{fig:protocol}
\end{figure}

We can see by inspection that these rules create the cycle described
above; Figure \ref{fig:protocol} follows an example with $n=3$
step-by-step through a full cycle.  A single cycle takes $n(2n+3)$
moves, so the full computation from the initial configuration to the
final one takes a total of $K = n(2n+3)(R-1)+n-1$ steps.  (The
additional $n-1$ steps are for the gate flag to move through the
final block.)

The only tricky part is to note the following fact:

\begin{fact}\label{fact:onefuture} For any configuration
containing only one active site, there is at most one transition
rule that applies to move forward one time step and at most one
transition rule to move backwards in time by one step. For {\em
valid} configurations, there is always {\em exactly} one forwards
transition and exactly one backwards transition, except for the
final and initial states, which have no future and no past,
respectively.
\end{fact}

\begin{proof}
We can see this by noting that the transition rules only refer to
the type of active site present and to the state of the site to the
right (for \gateflag, \rflag, and \turn) or left (for \lflag) of the
active site (or left and right respectively for backwards
transitions). For a configuration with only one active site, there
is therefore at most one transition rule that applies. We can also
see that for every legal configuration, there is indeed a transition
rule that applies; the few cases that may appear to be missing
cannot arise because of the constraints on alignment of the string
of qubits with the blocks.
\end{proof}
The illegal cases without forward transitions are \gateflag \block
\waitr, \gateflag \unborn, \rflag \block \unborn, and \rflag \block,
and the illegal cases without backwards transitions are \waitl
\block \gateflag, \dead \gateflag, \dead \block \rflag, and \block
\rflag.  These arrangements will play an important role in
Section~\ref{sec:QMA}.

{\em Remark:} For adiabatic computation, we can simplify this
construction a bit. For the proof of $\QMA$-completeness, we will
need to distinguish the particles to the left of the active site
from the particles to its right, but this is not necessary for
universal adiabatic computation. Therefore, we may combine the
\waitl\ and \waitr\ states into a single qubit subspace \waitr, and
we may combine \dead\ and \unborn\ into a single state \unborn. That
leaves us with essentially the same construction, but with only $9$
states per particle instead of $12$.

\section{Universality of adiabatic evolution in one dimension}\label{sec:adiabatic}
We first recall the definition of adiabatic computation \cite{farhiad,farhi}. A quantum system slowly evolves
from an easy to prepare ground state of some initial Hamiltonian $H_0$ to the ground state of a tailored
final Hamiltonian $H_1$, which encodes the problem. The state evolves according to Schr\"odinger's equation
with a time-varying Hamiltonian $H(s)= (1-s)H_0+sH_1$, with $s=t/T$, where $t$ is the time and $T$ is the
total duration of the adiabatic computation. Measuring the final ground state in the standard basis gives the
solution of the problem. The adiabatic theorem tells us that if $T$ is chosen to be large enough, the final
state will be close to the ground state of $H_1$. More specifically, we can define instantaneous Hamiltonians
$H(t/T)$, and let $g_{\rm min}$ be the minimum of the spectral gaps over all the Hamiltonians $H(t/T)$. $T$
is required to be polynomial in $1/g_{\rm min}$.

To prove Theorem \ref{thm:adiabatic}, given a general quantum
circuit, we would like to design an efficient simulation of this
circuit by adiabatic evolution on a one-dimensional line. The
construction in Section \ref{sec:protocol} associates to any circuit
$C_x$ with $n$ qubits a modified circuit $\Cmod{x}$ using $nR$
$9$-state particles in one dimension (with some states merged
according to the remark in the end of Sec. \ref{sec:protocol}),
together with transition rules which evolve the system for
$K=n(2n+3)R+n-1$ steps through $K+1$ configurations from initial to
final. We denote the quantum state of the $t$-th configuration by
$\ket{\gamma(t)}$ for $t=0,\ldots ,K$. Observe that the
$\ket{\gamma(t)}$ are orthogonal to each other, since each has a
different configuration. We note that the last state
$\ket{\gamma(K)}$ contains the final state of the original circuit,
encoded in the last block of particles.

Next we construct a one-dimensional adiabatic evolution whose final
Hamiltonian has as its ground state the history state of this
evolution, namely,
\[ \frac{1}{\sqrt{K+1}}\sum_{t=0}^{K} \ket{\gamma(t)}.\]
To define $H_1 \equiv \Hprop$, we simply translate the transition
rules of the previous section to \kdlocal{2}{9} Hamiltonians. In
other words, any transition rule which takes a state $|\alpha\ra$ to
$|\beta\ra$ is translated to a Hamiltonian of the form $\frac{1}{2}
(\ketbra{\alpha}{\alpha} + \ketbra{\beta}{\beta} -
\ketbra{\alpha}{\beta} - \ketbra{\beta}{\alpha})$. E.g., the first
rule becomes $\frac{1}{2}(\ketbra{\gateflag \waitr}{\gateflag
\waitr} + U \ketbra{\waitr \gateflag}{\waitr \gateflag} U^\dagger -
U \ketbra{\waitr \gateflag}{\gateflag \waitr} - \ketbra{\gateflag
\waitr}{\waitr \gateflag} U^\dagger)$, summed over a basis of states
for the encoded qubits and over all nearest neighbor pairs of
particles (recall we have combined \waitr\ and \waitl).  Here $U$ is
the gate corresponding to this location and acts on the two qubits
encoded as $\ket{\waitr \gateflag}$.  Restricted to the
eight-dimensional subspace spanned by $\ket{\gateflag \waitr}$ and
$\ket{\waitr \gateflag}$ (each for values of $\ket{00}$, $\ket{01}$,
$\ket{10}$, and $\ket{11}$ for the two qubits) for a single pair of
particles, that would give the matrix
\begin{equation*}
\frac{1}{2} \left( \begin{array}{cc}%
I & -U \\
-U^\dagger & I
\end{array} \right),
\end{equation*}
with $I$ the $4 \times 4$ identity matrix.

Our initial Hamiltonian has the initial configuration
$\ket{\gamma(0)}$ as its ground state. To define $H_0$, we penalize
all configurations that do not have
\gateflag\ in its $\ket{0}$ state ($\ket{\gateflag(0)}$) in the
first position:
\[
 H_0 = I - \ket{\gateflag(0)}\bra{\gateflag(0)}_1
%
 \]

This Hamiltonian of course has a highly degenerate ground state, but the important point is that
$\ket{\gamma(0)}$ is the only state which satisfies $H_0$ in the invariant subspace spanned by the
$\ket{\gamma(t)}$.  We assume we are able to select the correct ground state $\ket{\gamma(0)}$ to be the
actual initial state of the adiabatic computation.  We could avoid the need for this by various methods, for
instance by having an additional phase of adiabatic evolution from a simple non-degenerate Hamiltonian to
$H_0$.  Another option is to break the degeneracy of the ground state of $H_1$ by using the full $12$-state
construction with the Hamiltonian described in Section~\ref{sec:QMA}.  Even then, we either need to modify
$H_0$ to be non-degenerate (this approach requires a new analysis of the spectral gap, which we have not
done), or to add an additional $13$th start state. We outline the latter approach in Section~\ref{sec:nondegenerate}.

To analyze the spectral gap of any Hamiltonian in the convex
combination of $H_0$ and $H_1$, we essentially follow the proof of
\cite{ad1}, using an improvement due to \cite{Ruskai:adimprove}. We
first observe that
\begin{claim}
The subspace $\calK_0$ spanned by the states $\ket{\gamma(t)}$ is
invariant under any of the Hamiltonians $H(s)$.
\end{claim}
This is easily seen to be true by the fact that it is invariant
under the initial and final Hamiltonians. From this claim, it
follows that throughout the evolution we are {\em inside} $\calK_0$,
so we only need to analyze the spectral gap in this subspace. We
next move to the basis of $\calK_0$ spanned by the $\ket{\gamma(t)}$
states. In this basis $H_0$ restricted to $\calK_0$ looks like
\begin{eqnarray*}\label{eq:h_initial_s0}
H_0 |_{\calK_0} & = & \left(%
\begin{array}{cccc}
  0 &      0 & \ldots      & 0 \\
  0 &      1 & \ldots      & 0 \\
  \vdots & \vdots & \ddots & \vdots \\
  0 &      0 & \ldots & 1  \\
\end{array}%
\right),
\end{eqnarray*}
and $H_1$ restricted to $\calK_0$ becomes
\begin{eqnarray*}\label{eq:h_final_s0}
H_1 |_{\calK_0} & = &  \left(
\begin{array}{rrrrrrr}
\smfrac{1}{2} & \mns \smfrac{1}{2} &0 & & \cdots& & 0 \\ \mns \smfrac{1}{2} & 1 & \mns \smfrac{1}{2} & 0 &
\ddots & & \vdots\\ 0 & \mns \smfrac{1}{2} & 1 & \mns \smfrac{1}{2} & 0 & \ddots & \vdots\\ & \ddots & \ddots
& \ddots & \ddots & \ddots & \\ \vdots& & 0 & \mns \smfrac{1}{2} &1 & \mns \smfrac{1}{2}& 0 \\ & & & 0 & \mns
\smfrac{1}{2} &1 & \mns \smfrac{1}{2} \\ 0& & \cdots& & 0&
\mns \smfrac{1}{2} & \smfrac{1}{2} \\
\end{array}
\right).
\end{eqnarray*}
It can easily be seen that the history state (which in this basis is
simply the all ones vector) is a zero eigenstate of this
Hamiltonian. Finally we need to analyze the spectral gaps of all
convex combinations of these two. This has been done in Sec.~3.1.2
of \cite{ad1} and simplified with improved constants in
\cite{Ruskai:adimprove}; we refer the reader there for the details.
The result is that the spectral gap is at least $1/[2(K+1)^2]$, an
inverse polynomial in $n$ and $R$, which itself is an inverse
polynomial in the number of gates in the original circuit. This
proves Theorem \ref{thm:adiabatic}.

\section{1D QMA}
\label{sec:QMA}

The propagation Hamiltonian $\Hprop$ introduced in
Section~\ref{sec:adiabatic} is insufficient for $\QMA$-completeness.
Now we have a circuit $C_x$ which checks the witness for a $\QMA$
problem, taking as input the witness and some ancilla qubits in the
state $\ket{0}$. However, any correct history state for the circuit
$C_x$ will have zero eigenvalue for $\Hprop$, even if the $t=0$
component of the history state is not correctly initialized, or if
$C_x$ does not accept the witness. Even worse, $\Hprop$ also has
zero eigenvalue for any other state which is a uniform superposition
of states connected by the transition rules, even if the
superposition includes illegal configurations.

To solve these problems, we will introduce three new terms to the
overall Hamiltonian $H$:
\begin{equation*}
H = \Hprop + \Hinit + \Hfinal + \Hpenalty.
\end{equation*}
The initialization term $\Hinit$ will constrain the initial state of
the modified checking circuit $\Cmod{x}$ so that all ancilla qubits
are initialized to $\ket{0}$, and the final Hamiltonian $\Hfinal$
verifies that the checking circuit does accept the witness as input.
$\Hpenalty$ will penalize illegal configurations using local
constraints. As mentioned in the introduction, not all illegal
configurations can be penalized directly. Instead, some of them will
only be penalized because they evolve into a locally checkable
illegal configuration.

To create the initialization term, we will assume without loss of
generality that all the gates performed in the first block are the
identity. We will use the gate flag \gateflag\ to check that qubits
are correctly initialized instead of using it to do gates in the
first block. Then we get the following Hamiltonian term:
\begin{equation*}
\Hinit = \sum_i \ketbra{\mbox{\gateflag (1)}}{\mbox{\gateflag
(1)}}_i.
\end{equation*}
The sum is taken over ancilla qubits $i$ in their starting positions
in the first block of $n$ sites.  (The remaining qubits in the first
block of $n$ sites are used to encode the potential witness.)
$\Hinit$ creates an energy penalty for any ancilla to be in the
state $\ket{1}$ when the gate flag passes over it.  Since the gate
flag sweeps through the whole block, this ensures that the ancilla
qubits must be correctly initialized to $\ket{0}$ or the state
suffers an energy penalty.

Similarly, for the final Hamiltonian, we assume no gates are performed in the final block, and use \gateflag\
to check that the circuit accepts the output.  That is, we get the following term in the Hamiltonian:
\begin{equation*}
\Hfinal = \ketbra{\mbox{\gateflag (0)}}{\mbox{\gateflag (0)}}_{\rm
out}
\end{equation*}
This causes an energy penalty if the output qubit ``out'' is in the
state $\ket{0}$ when the gate flag sweeps over it.  In general, a
correct history state for some potential input witness state will
have, for the final block, a superposition of terms with the output
qubit in the state $\ket{0}$ and terms with the output qubit in the
state $\ket{1}$, and the energy penalty is thus proportional to the
probability that the circuit rejects the potential witness.

Now we move to $\Hpenalty$, which is a bit more involved.  We will
describe a set of local penalties that will enforce that each
illegal configuration will be penalized either directly or because
it will evolve into a configuration that will be penalized. We
forbid the following arrangements:
\begin{enumerate}

\item \unborn X, \unborn \block X (X is anything but \unborn ), X\dead,
X\block \dead\  (X is anything but \dead )

\item In the first block on the left: \unborn , \block \waitr . In the last block on the right: \dead , \waitl \block

\item \waitr X, \waitr \block X (X is anything but \waitr\ or \unborn),
X\waitl, X\block \waitl\ (X is anything but \waitl\ or \dead )

\item \dead \waitr, \dead \block \waitr , \waitl \unborn , \waitl \block \unborn , \dead
\unborn , \dead \block \unborn

\item \waitl \waitr , \waitl \block \waitr

\item Any adjacent pair of active sites (e.g., \gateflag \rflag\ or \lflag
\rflag), with or without block boundaries

\item \dead \gateflag, \gateflag \unborn,
\waitl \block \gateflag, \gateflag \block \waitr, \dead \block
\rflag, \rflag \block \unborn, \block \rflag, \rflag \block\ \ (but
the first two are OK with a block boundary and the next four are OK
without a block boundary)

\end{enumerate}
Note that the arrangements forbidden in group 7 are precisely those
missing a forward or backwards transition rule in the note following
Fact~\ref{fact:onefuture}.

We encode these rules into a penalty Hamiltonian in the
straightforward way:
\begin{equation*}
\Hpenalty = \sum_{XY} \ketbra{XY}{XY}_{i,i+1},
\end{equation*}
where the sum is taken over the forbidden arrangements XY listed
above for all adjacent pairs of sites (tailored appropriately to the
location of block boundaries and the first and last blocks).

\begin{claim}\label{cl:legal_non_detectable}
A configuration that satisfies the rules in groups 1 through 6 is of
one of the legal forms described in Section~\ref{sec:protocol}, or
it is in one of three cases: (i) a ({\bf qubits}) string of the form
\waitl \cdotiki \waitl \turn \waitr \cdotiki \waitr\ (such a
configuration is allowed in Section~\ref{sec:protocol} only when
\turn\ is all the way on one end of the ({\bf qubits}) string), (ii)
a ({\bf qubits}) string that has incorrect length: different from
$n$ when the active site is \gateflag\ or \rflag\ or different from
$n+1$ when the active site is \lflag\ or \turn, or (iii) a ({\bf
qubits}) string of length $n$, but with either an active site
\rflag\ and the ({\bf qubits}) string aligned with a block boundary,
or an active site \gateflag\ and the ({\bf qubits}) string {\em not}
aligned with a block boundary.
\end{claim}

Note that we do not need the rules in group 7 for this claim; they
are used later to deal with the illegal states that cannot be
checked locally.

\begin{proof}
The rules in group 1 enforce that the configuration is of the form
\dead \cdotiki \dead ({\bf qubits})\unborn \cdotiki \unborn\ by
ensuring that all unborn (\unborn) states are at the far right and
all dead (\dead) states are at the far left.  Group 2 ensures that
the particles are not all \unborn\ or all \dead\ (and in fact
ensures that there is at least one block of $n$ composed of
something else). Group 3 ensures that within the ({\bf qubits})
string, any \waitr\ qubits are on the right end and any \waitl\
qubits are on the left end. Groups 2 and 4 guarantee that the ({\bf
qubits}) string is not empty and does not consist of only \waitl\ or
\waitr\ qubits. Group 5 then checks that the ({\bf qubits}) string
does not directly jump from \waitl\ to \waitr, so there is at least
one active site in between, and the rules in group 6 ensure that
there cannot be more than one adjacent active site.  Since, by
groups 1 and 3, an active site can only occur to the right of all
\dead\ and \waitl\ sites and to the left of all \waitr\ and \unborn\
sites, all active sites are together in a group, so the rules in
groups 5 and 6 ensure that there is exactly one active site.
Comparing with Definition~\ref{def:legal}, we see that this leaves
only the exceptions listed in the claim.
\end{proof}

The three exceptions in the above claim cannot be ruled out directly
via a local check, since counting cannot be done using local
constraints.%
\footnote{If we are willing to add a $13$th state, we could split
the \turn\ state into two states, one to turn around on the left and
one to turn around on the right. Then we could rule out \waitl
\cdotiki \waitl \turn \waitr \cdotiki \waitr\ with local rules.
However, this strategy will not work for strings with the wrong
number of qubit states.}
We can only rule out the exceptions by considering both the penalty
Hamiltonian $\Hpenalty$ and the propagation Hamiltonian $\Hprop$.

In order to do so, we first break down the full Hilbert space into
subspaces which are invariant under both $\Hpenalty$ and $\Hprop$.
Let us consider the minimal sets of configurations such that the
sets are invariant under the action of the transition rules. This
defines a partition of the set of configurations. Given such a
minimal set $S$, we consider the subspace $\calK_S$ spanned by all
the configurations in $S$. The penalty Hamiltonian is diagonal in
the basis of configurations, so $\calK_S$ is invariant under
$\Hpenalty$, and the set $S$ is closed under the action of $\Hprop$,
so $\calK_S$ is invariant under $\Hprop$ as well. The space
$\calK_S$ belongs to one of three types:

\begin{enumerate}
\item All states in $S$ are legal.  (We call this subspace
$\calK_0$.)
\item All states in $S$ are illegal, but are all locally detectable, namely, none of them belong to the
exceptions of Claim \ref{cl:legal_non_detectable}.
\item All states in $S$ are illegal, but at least one of them is not locally detectable.
\end{enumerate}

Note that the set $S$ cannot contain both legal and illegal
configurations, since legal configurations do not evolve to illegal
ones, and vice-versa.  We want to show that subspaces of type 2 or 3
have large energy.  The main challenge will be to give a lower bound
on the energy of a space $\calK_S$ of type 3, which we will do by
lower bounding the fraction of configurations in $S$ that violate
one of the conditions 1-7 above. We want to prove:

\begin{lemma}[Clairvoyance lemma]
\label{lemma:clair} The minimum eigenvalue of $\Hprop+\Hpenalty$,
restricted to any $\calK_S$ of type $2$ or $3$, is $\Omega(1/K^3)$,
where $K$ is the number of steps in $\Cmod{x}$.
\end{lemma}

\begin{proof}
To prove the theorem, we deal separately with each type of invariant
subspace. Type $2$ is straightforward: the full subspace $\calK_S$
has energy of at least $1$, due to $\Hpenalty$.

We now focus on type $3$. We observe several properties of type $3$
subspaces. First, notice that the number of qubit states (\waitl,
\waitr, \gateflag, and \rflag) and the number of active sites
(\turn, \gateflag, \rflag, and \lflag) are conserved under all
transition rules. Since $S$ is a minimal set preserved by the
transition rules, all the configurations in $S$ therefore contain
the same number of qubit states. Since $S$ is of type $3$, it must
contain an undetectable illegal configuration, such as a ({\bf
qubits}) string of the wrong length, and by
Claim~\ref{cl:legal_non_detectable} has exactly one active site
somewhere in the string. %
By Fact~\ref{fact:onefuture}, there is at most one forwards
transition rule and one backwards transition rule that applies to
each configuration in $S$. Note also that configurations with one
active site which is \lflag\ or \turn\ always have both a forward
and backwards transition rule available.  Therefore, if $S$ is of
type $3$ it must contain at least one configuration with a \rflag\
or \gateflag\ active site.

Going further, we claim that a fraction of at least $\Omega(1/n^2
R)$ of the configurations in $S$ violate one of the rules in groups
1 through 7 above. To see this, we divide into three cases:
\begin{itemize}
\item $S$ contains a configuration with a qubit string of the form \waitl \cdotiki \waitl \turn \waitr
\cdotiki \waitr: In this case, the forward transition rule that
applies to the illegal configuration is either \turn \waitr\
$\rightarrow$ \dead \rflag\ or \turn \block \waitr\ $\rightarrow$
\dead \block \gateflag.  This gives us a configuration which can be
locally seen as illegal, such as \waitl \cdotiki \waitl \dead \rflag
\waitr \cdotiki \waitr. This violates group 1, as \dead\ sites
should be to the left of all other kinds of sites. Therefore $S$
contains at least $1/2$ locally checkable illegal configurations. In
fact, it is a much larger fraction, as all the backwards transition
rules and further forward transition rules will also produce locally
checkable illegal states.

\item $S$ contains a configuration with a \gateflag\ active site:
Via forward transitions, the \gateflag\ will move to the right until
it hits either the right end of the ({\bf qubits}) string or until
it hits a block boundary.  If it does not encounter both together,
it has reached one of the illegal configurations forbidden in group
7.  Via backwards transitions, the \gateflag\ moves to the left
until it hits either the left end of the ({\bf qubits}) string or
until it hits a block boundary.  Again, if it does not encounter
both together, we must now be in one of the configurations forbidden
in group 7.  If the right and left ends of the ({\bf qubits}) string
are both lined up with block boundaries, and there are no other
block boundaries in between, then the ({\bf qubits}) string contains
exactly $n$ sites and is correctly aligned with the block boundary.
Therefore it is actually a legal state, which cannot occur in $S$ of
type $3$.  Thus, we must have one of the locally checkable illegal
states, and it must occur after at most $n$ \gateflag\ transitions
(in either direction).  The \gateflag\ configurations themselves,
however, could be part of a larger cycle with transitions involving
\lflag, \rflag, and \turn\ active sites.  However, the length of the
({\bf qubits}) string is at most $nR$, and we can have at most $n$
\lflag-\turn-\rflag-\turn\ cycles before getting a \gateflag, so the
locally checkable illegal states comprise a fraction at least
$\Omega(1/n^2 R)$ of all the states in $S$.

\item $S$ contains no configuration with a \gateflag\ active site:
$S$ still must contain configurations with \rflag.  These
configurations can form part of a normal cycle, but the normal cycle
eventually leads to a transition to a \gateflag\ configuration.  For
forward transitions, this occurs when the left end of the ({\bf
qubits}) string is aligned with a block boundary, and for backwards
transitions, it occurs when the right end of the ({\bf qubits})
string is aligned with a block boundary.  The only way to avoid this
happening is for $S$ to have a \rflag\ configuration with no forward
transition and another with no backwards transition.  Since the only
such configurations contain one of the arrangements \dead \block
\rflag, \rflag \block \unborn, \block \rflag, or \rflag \block, all
of which are forbidden by the rules in group 7, $S$ must contain a
locally checkable illegal configuration.  Again, the length of the
({\bf qubits}) string is at most $nR$, and we can have at most $n$
\lflag-\turn-\rflag-\turn\ cycles in a row, so the locally checkable
illegal states comprise a fraction at least $\Omega(1/n^2 R)$ of all
the states in $S$.

\end{itemize}

We can now invoke Lemma~$14.4$ from \cite{Kitaev:book} to lower
bound the energy of the overall Hamiltonian for a type $3$ subspace:

\begin{lemma}
\label{lemma:Kitaev} Let $A_1$, $A_2$ be nonnegative operators, and
$L_1$, $L_2$ their null subspaces, where $L_1 \cap L_2 = \{0\}$.
Suppose further that no nonzero eigenvalue of $A_1$ or $A_2$ is
smaller than $v$. Then
\begin{equation*}
A_1 + A_2 \geq v \cdot 2 \sin^2 {\theta/2},
\end{equation*}
where $\theta = \theta (L_1, L_2)$ is the angle between $L_1$ and
$L_2$.
\end{lemma}

In our case, $A_1$ is the propagation Hamiltonian $\Hprop$, and its
null subspace, restricted to $K_S$, consists of equal superpositions
over all configurations in the invariant subspace $S$. (There are
multiple such states, with different values of the encoded qubit
states.)  $A_2$ is the penalty Hamiltonian $\Hpenalty$, diagonal in
the basis of configurations. Then $\sin^2 \theta$ is the projection
(squared) of the superposition of all shapes on the subspace of
locally checkable illegal configurations; that is, it is the
fraction of locally checkable illegal configurations in the
invariant set.  The minimum nonzero eigenvalue of $\Hpenalty$ is
$1$, but (as in~\cite{Kitaev:book}) the minimum nonzero eigenvalue
of $\Hprop$ is $\Omega(1/K^2)$, where $K = n(2n+3)(R-1)+n-1$ is the
number of steps in $\Cmod{x}$. Thus, if $S$ is a set containing a
configuration which is illegal but cannot be locally checked, all
states in $K_S$ have an energy at least $\Omega(1/K^3)$.
\end{proof}

We can now prove Theorem \ref{thm:qma}. We start by assuming the
circuit $C_x$ accepts correct witnesses and rejects incorrect
witnesses with a probability exponentially close to $1$. This can be
achieved, for instance, by checking multiple copies of the witness.
Then we will show that if there exists a witness which is accepted
by $C_x$ with probability at least $1-O(1/K^{3})$, then there is a
state with energy at most $O(1/K^{4})$, whereas if all possible
witnesses are only accepted by $C_x$ with probability at most
$1/K^{3}$, then all states have energy at least $\Omega(1/K^3)$.

We already know that on subspaces of type $2$ and $3$, the minimum
eigenvalue of $\Hprop+\Hpenalty$ is $\Omega(1/K^3)$.  We therefore
restrict attention to the subspace $\calK_0$ of type $1$, built from
only legal configurations, and can hence ignore $\Hpenalty$.  From
this point on, the proof follows~\cite{Kitaev:book}, but we include
it for completeness.

The type $1$ space $\calK_0$ contains only valid history states.
However, some have incorrect ancilla inputs and are penalized by
$\Hinit$, while others include an incorrect witness input, and are
penalized by $\Hfinal$. Only history states corresponding to a
correct witness have low energy.

If there exists a witness which is accepted by $C_x$ with
probability at least $1-1/K^{3}$, then the history state $\ket{\Phi}
= \frac{1}{\sqrt{K+1}}\sum_t \ket{\phi_t}$ for that witness will have small enough
energy:
\begin{equation*}
\bra{\Phi}\Hprop + \Hpenalty + \Hinit \ket{\Phi} = 0,
\end{equation*}
so we only need to calculate $\bra{\Phi} \Hfinal \ket{\Phi}$. The
only term in $\ket{\Phi}$ that contributes is the one with
\gateflag\ on the output qubit in the last block, which has a
coefficient $1/\sqrt{K+1}$.  Therefore, the energy due to $\Hfinal$
is $1/(K+1)$ times the probability that the witness is rejected by
$C_x$. That is, it is $O(1/K^{4})$.

Now let us consider what happens for a ``no'' instance.  The
following lemma will complete the proof of Theorem \ref{thm:qma}:

\begin{lemma}
\label{lemma:noinstance}
If there is no witness that causes the circuit to accept with
probability larger than $1/K^{3}$, then every state in $\calK_0$ has
energy at least $\Omega(1/K^3)$.
\end{lemma}

\begin{proof}
We will invoke Lemma~\ref{lemma:Kitaev} once more, with $A_1 =
\Hprop$, and $A_2 = \Hinit + \Hfinal$. From now on, all the
operators we discuss will be restricted to $\calK_0$. The null
subspace of the propagation Hamiltonian $\Hprop$ contains states
which correspond to a history $\ket{\psi} =
\frac{1}{\sqrt{K+1}}\sum_t \ket{\phi_t}$ of the circuit $\Cmod{x}$
for some input state $\ket{\phi_0}$, not necessarily correct. Fix
such a history state $\ket{\psi}$. We will upper bound its
projection squared on the null space of $A_2$.

Let us first study the structure of this null space. The null space
of $\Hinit$ is spanned by states for which either the \gateflag\ is
not on an ancilla qubit in the first block or the \gateflag\ is on
an ancilla qubit, but the ancilla qubit is a $\ket{0}$. Let
$\Pi_{\rm init}$ be the projection on the null space of $\Hinit$.
The null space of $\Hfinal$ contains states for which either the
\gateflag\ is not on the output qubit in the last block or the
\gateflag\ is on the output qubit and the output qubit is a
$\ket{1}$. Let $\Pi_{\rm final}$ be the projection on the null space
of $\Hfinal$. The two projectors commute, so the projector onto the
null subspace of $\Hinit + \Hfinal$ is $\Pi_{\rm init} \Pi_{\rm
final}$. We would thus like to upper bound
\[\bra{\psi} \Pi_{\rm init} \Pi_{\rm final} \ket{\psi}. \]

To this end, we expand the input state $\ket{\phi_0}$ for
$\ket{\psi}$ into $\ket{\phi_0} = \alpha \ket{\phi^{v}_0} + \beta
\ket{\phi^{i}_0}$, where $\ket{\phi^{v}_0}$ is a state with all
valid ancilla input qubits and $\ket{\phi^{i}_0}$ is an orthogonal
state consisting of a superposition of states with one or more
invalid ancilla input qubits. Note that $\ket{\psi} = \alpha
\ket{\psi^v} + \beta \ket{\psi^i}$, where $\ket{\psi^v}$ is the
history state for input $\ket{\phi^{v}_0}$ and $\ket{\psi^i}$ is the
history state for input $\ket{\phi^{i}_0}$. We have:
\begin{equation*}
\bra{\psi} \Pi_{\rm init} \Pi_{\rm final} \ket{\psi} = |\alpha|^2
\bra{\psi^v} \Pi_{\rm init}\Pi_{\rm final} \ket{\psi^v} +
|\beta|^2\bra{\psi^i} \Pi_{\rm init} \Pi_{\rm final} \ket{\psi^i} +
2 {\rm Re} \left(\alpha^* \beta \bra{\psi^v} \Pi_{\rm init}\Pi_{\rm
final} \ket{\psi^i} \right),
\end{equation*}
where we have used the fact that the two projectors commute, so
$\Pi_{\rm init}\Pi_{\rm final}$ is Hermitian.

We now bound each of the three terms. For the first term, note that
$\bra{\psi^v} \Pi_{\rm init} = \bra{\psi^v}$. We might as well
assume the output qubit is the very last qubit, so the \gateflag\
reaches it only at $t=K$, and we have that
\begin{equation*}
\bra{\psi^v} \Pi_{\rm final} \ket{\psi^v} = \frac{1}{K+1} \left[ K +
\bra{\phi^v_K} \Pi_{\rm final} \ket{\phi^v_K} \right] = \frac{K +
p}{K+1},
\end{equation*}
where $p$ is the probability that the circuit $C_x$ accepts
$\ket{\phi^v_0}$. We thus find that the first term is at most
$|\alpha|^2 (K + O(1/K^3))/(K+1)$.

For the second term, it suffices to upper bound $\bra{\psi^i}
\Pi_{\rm init} \ket{\psi^i}$ since
\begin{equation*}
\bra{\psi^i} \Pi_{\rm init} \ket{\psi^i} = \| \Pi_{\rm init}
\ket{\psi^i}\|^2 \ge \|\Pi_{\rm final} \Pi_{\rm init}
\ket{\psi^i}\|^2 = \bra{\psi^i} \Pi_{\rm init} \Pi_{\rm final}
\ket{\psi^i},
\end{equation*}
again using the fact that the two projectors commute. Since the
computation $\Cmod{x}$ begins by sweeping a \gateflag\ through the
first $n$ sites, and thereafter never has \gateflag\ within the
first $n$ sites, only the first $n$ states $\ket{\phi_t}$ can
possibly be outside the null space of $\Hinit$.  Also, note that
$\bra{\phi_t} \Pi_{\rm init} \ket{\phi_{t'}} = 0$ unless $t = t'$.
We can thus write
\begin{equation*}
\bra{\psi^i} \Pi_{\rm init} \ket{\psi^i} = \frac{1}{K+1}
\left[(K+1-n) + \sum_{t=0}^{n-1} \bra{\phi^i_t} \Pi_{\rm init}
\ket{\phi^i_t} \right].
\end{equation*}
We have $\bra{\phi_t} \Pi_{\rm init} \ket{\phi_t} = 0$ if the
$t^{\rm th}$ qubit of the input is an ancilla qubit and is in the
state $\ket{1}$, and $\bra{\phi_t} \Pi_{\rm init} \ket{\phi_t} = 1$
otherwise. Since the $\ket{\phi^i_t}$ states are superpositions of
terms with at least one incorrect ancilla input qubit, we have
\begin{equation*}
\bra{\psi^i} \Pi_{\rm init} \ket{\psi^i} \leq K/(K+1).
\end{equation*}

For the third term, we have
\begin{eqnarray*}
|{\rm Re} \left(\alpha^* \beta \bra{\psi^v} \Pi_{\rm init}\Pi_{\rm
final} \ket{\psi^i}\right)| \le & |\bra{\psi^v} \Pi_{\rm init}
\Pi_{\rm final} \ket{\psi^i}| \\
= & |\bra{\psi^v}\Pi_{\rm final} \ket{\psi^i}|.
\end{eqnarray*}
Note, however, that $\braket{\phi^v_t}{\phi^i_t} = 0$, since
$\ket{\phi_t}$ is related to $\ket{\phi_0}$ by a unitary
transformation and $\braket{\phi^v_0}{\phi^i_0} = 0$.  This gives us
\begin{eqnarray*}
|\bra{\psi^v} \Pi_{\rm final} \ket{\psi^i}| = & \frac{1}{K+1}
|\bra{\phi^v_K} \Pi_{\rm final} \ket{\phi^i_K}| \\
\leq & \frac{1}{K+1} \| \Pi_{\rm final} \ket{\phi^v_K} \| \\
= & O(K^{-5/2}),
\end{eqnarray*}
The above expression might not be $0$ since both $\ket{\psi^v}$ and
$\ket{\psi^i}$ might have a non-zero component which is accepted by
the circuit, but since $\ket{\phi^v}$ is only accepted rarely (with
probability $O(1/K^3)$), it cannot be too large.

Summing up all contributions, we have:
\begin{equation*}
\bra{\psi} \Pi_{\rm init} \Pi_{\rm final} \ket{\psi} \leq |\alpha|^2
\frac{K + O(1/K^3)}{K+1} + |\beta^2| \frac{K}{K+1} + O(K^{-5/2}) =
\frac{K}{K+1} + O(K^{-5/2}).
\end{equation*}
From this we can conclude that $\sin^2(\theta)$ is at least
$\Omega(1/K)$, where $\theta$ is the angle between the null spaces
of $\Hprop$ and $\Hinit + \Hfinal$.  Lemma~\ref{lemma:Kitaev} then
tells us that every state has energy $\Omega(1/K^3)$.
\end{proof}

\subsection{Non-degenerate Universal Adiabatic Hamiltonian}
\label{sec:nondegenerate}

If we wish to create a universal adiabatic quantum computer with a
non-degenerate ground state, the Hamiltonians from Section~\ref{sec:adiabatic}
 do not suffice. In fact, even if we used a $12$-state construction with
the final Hamiltonian $H_1$ equal to $H$ from the $\QMA$-completeness
 construction above, then $H_1$ would be non-degenerate, but the
 initial Hamiltonian $H_0$ in Section~\ref{sec:adiabatic} still is not.
We sketch here a construction that adds an additional
$13$th state to make sure that the Hamiltonian remains non-degenerate
 throughout the adiabatic evolution.

The new state is a start state \start, which serves as the state for the input qubits in the initial configuration. We will add a
set of terms to $H_0$ and $H_1$ (which also includes $\Hinit$, $\Hfinal$, and $\Hpenalty$, in addition to $\Hprop$). These terms will guarantee that the unique ground state of $H_0$ is \block
\gateflag \start \cdotiki \start \block \unborn \cdotiki \unborn \block \cdotiki \block \unborn \cdotiki
\unborn \block. Furthermore, $H_0$ will have the property that $H_0 \ket{\gamma(t)} = 1$ for any $t > 0$, which allows us to use the calculation of the gap given in Section~\ref{sec:adiabatic}.
To do this, we add a term that enforces the condition that if the first particle is in state \gateflag, then
the second particle is in state \start. This is accomplished by forbidding any configuration that has
\gateflag\ in the first location and any state except \start\ in the second. For example, to forbid states
\gateflag\unborn\ in particles 1 and 2, we add the term $\ket{\gateflag\unborn}\bra{\gateflag\unborn}_{12}$.
Similarly, for $i = 2, \ldots, n-1$, we enforce that if particle $i$ is in state \start\ then particle $i+1$
must be in state \start\ as well. Then we add a term that says if particle $n$ is in state \start\, then
particle $n+1$ is in state \unborn. For $i > n$, add terms that enforce that if particle $i$ is in
state \unborn, then so is particle $i+1$. Similar terms should also be added to $H_1$ to assure that its ground state remains non-degenerate when the new start state is added.  Finally, we need to slightly alter the transition rules used to build $\Hprop$. For $i = 1, \ldots, n-1$, replace the rule \gateflag \waitr\ $\rightarrow$ \waitl
\gateflag\ in locations $i$ and $i+1$ with \gateflag \start\ $\rightarrow$ \waitl \gateflag(0). With this
change, Fact \ref{fact:onefuture}  still holds, with the small caveat that valid configurations now must have any qubits in the first block set to $\ket{0}$. Note that none of the configurations in the subspace spanned
by the $\ket{\gamma(t)}$ violate any of the extra conditions. Furthermore, the additional terms guarantee the
uniqueness of the ground state of $H_0$ given above. In order to establish the non-degeneracy of any of the
Hamiltonians $H(s)$, we need to prove that there are no low energy states outside of the subspace spanned by
the $\ket{\gamma(t)}$, which follows from the Clairvoyance lemma and a simplified version of Lemma~\ref{lemma:noinstance}.

\section{Discussion and Open Problems}\label{sec:conclude}

We have shown that \rDHam{1}{12} Hamiltonians can be used for both
universal adiabatic quantum computation and to produce very
difficult, $\QMA$-complete problems.

A similar result holds also for a related quantum
complexity class, called $\QCMA$, which is the subclass of $\QMA$
where the witnesses are restricted to be classical.
To see this, observe first that our reduction from an
arbitrary $\QMA$ language to the \rDHAM{1}{12} problem is
witness-preserving, at least once the acceptance probability has
been amplified. Given a witness for the original $\QMA$ language, we
can, in fact, efficiently construct on a quantum computer a witness
for the corresponding instance of \rDHAM{1}{12} using the adiabatic
algorithm. This implies that if the witness for the original problem
is efficiently constructible (which means we may as well assume it
is a classical bit string describing the circuit used to construct
the quantum witness), then the witness for \rDHAM{1}{12} is also
efficiently constructible. Thus, we have also shown that the
sub-language of \rDHAM{1}{12} which has the additional promise of an
efficiently-constructible ground state is complete for $\QCMA$.

There remain many interesting related open problems.  Clearly, it is
interesting to ask whether the size of the individual particles in
the line can be further decreased, perhaps as far as qubits.  We
could, of course, interpret our $12$-state particles as sets of $4$
qubits, in which case our Hamiltonian becomes $8$-local, with the
sets of $8$ interacting qubits arranged consecutively on a line. The
perturbation-theory gadgets used to convert $8$-local interactions
to $2$-local ones do not work in a one-dimensional system: The pairs
of interacting qubits form a graph, which needs to have degree at
least $3$ (or $4$ for some of the gadgets used
by~\cite{Oliveira:05a}).  If we do apply the approach of
\cite{Oliveira:05a}, however, to the system described here, we get a
\rDqubit{2} Hamiltonian on a strip of qubits of constant width. The
constant is rather large, but this still constitutes an improvement
over the result of \cite{Oliveira:05a}.

Another approach to decreasing the dimensionality of the individual
particles is to find a new protocol for removing the explicit
reference to time from the circuit $C_x$. It seems likely that some
further improvement is possible in this regard, but it is unlikely
that an improved protocol by itself can take us all the way to
qubits, as we need additional states to provide the control
instructions.  If it is possible to have universal adiabatic quantum
computation and \QMA-completeness with \rDHam{1}{2} Hamiltonians,
we will probably need new techniques to prove it.  It may also be that there
is a transition at some intermediate number of states between $2$ and
$12$ for which universality and \QMA-completeness become possible.
This would be analogous to the classical $2$-dimensional Ising spin
problem without magnetic field, which is in P for a single plane of
bits, but is \NP-complete when we have two layers of
bits~\cite{barahona}.

Another interesting line of open questions is to investigate the
energy gap.  There are two energy gaps of relevance, both
interesting.  One is the ``promise gap'' $\Delta$ in the definition
of \QMA.  We have shown that we have \QMA-completeness when $\Delta$
is polynomially small relative to the energy per term. This can
easily be improved to a constant value of $\Delta$ by amplification:
$t$ copies of the ground state will either have energy less than
$tE$ or above $t(E+\Delta)$, amplifying the promise gap to
$t\Delta$.  A more interesting question is whether $\Delta$ can be
made a constant fraction of the total energy available in the
problem, the {\em largest} eigenvalue of $H$; if so, that would
constitute a quantum version of the PCP theorem. Hastings and
Terhal~\cite{HT} have argued that it is not possible
to do this in any constant number of dimensions unless $\QMA =
\Pclass$: To approximate the ground state energy, we can
divide the system into blocks of a constant size, and diagonalize
the Hamiltonian within each block. We can then consider the tensor
product of the ground states for each separate block; since such a
state ignores the energy due to Hamiltonian terms interacting
different blocks, it will not be the true ground state, but its
energy can differ from the ground state energy by at most the
surface area of each block times the number of blocks (normalizing
the maximum energy per term to be $1$), whereas the maximum
eigenvalue of $H$ is roughly proportional to the {\em volume} of each
block times the number of blocks. Thus, the tensor product state approximates
the ground state energy up to a constant fraction of the total
energy, and the fraction can be made arbitrarily small by increasing
the size of each block.  This gives a polynomial-time classical
algorithm for approximating the ground state energy to this
accuracy.  When the interactions in $H$ are not constrained by
dimensionality, the argument breaks down, so this quantum version of
PCP remains an interesting open problem for general \kdlocal{k}{r}
Hamiltonians.

We can also look at the spectral gap, the gap between the ground
state and the first excited state of $H$.  For adiabatic
computation, we are interested in the minimal spectral gap over the
course of the computation.  We have shown that universal adiabatic
quantum computation is possible if the spectral gap is polynomially
small relative to the energy per term. What happens if it is bounded
below by a constant? Hastings has recently shown~\cite{hastingsun}
that it is possible to efficiently classically simulate the
adiabatic evolution of a \rDHam{1}{r} Hamiltonian system with
constant spectral gap.  We should therefore not expect to be able to
perform universal adiabatic quantum computation with such
Hamiltonians, as that would imply $\BQP = \Pclass$.  Hastings'
argument builds on his earlier paper~\cite{hastings}, which showed
that the ground state of a gapped \rDHam{1}{r} Hamiltonian has an
efficient classical description as a matrix product state.  By
repeatedly updating the matrix product state description, one can
then keep track of the adiabatic evolution with resources polynomial
in the path length in parameter space and the inverse error of the
approximation.  For \rDHam{d}{r} Hamiltonians with constant spectral
gap and $d \geq 2$, the question remains open in general, although
Osborne has proven~\cite{osborne} that such systems can be
efficiently classically simulated for logarithmic time.

For the \QMA-completeness problem, we know very little about the
possible values of the spectral gap.  If the spectral gap is
constant relative to the energy per term (or even $\Omega(1/\ln\ln
nR)$), the ground state has a matrix product state
representation~\cite{hastings} and the problem is in \NP. In our
construction, there are enough low-energy states that the spectral
gap cannot be much larger than the promise gap
$\Delta$, but it might be much smaller.  In the ``no''
instances, there are likely very many states which violate only a
small number of transition rules, penalty terms, or initial
conditions, and these have energy just above $E + \Delta$, so likely
the spectral gap is exponentially small in the ``no'' instances.  We
know less about the size of spectral gap in the ``yes'' instances.
States which do not
correspond to valid histories have an energy at least $\Delta$
larger than the ground state, but unfortunately, there may be
different valid histories with energies less than $E+\Delta$ but
above the ground state energy.  The difficulty is that the original
problem might have many potential witnesses. Some may be good
witnesses, accepted with high probability, whereas others may be
mediocre witnesses, accepted with a probability near $1/2$. There
could, in fact, be a full spectrum of witnesses with only
exponentially small gaps between their acceptance probabilities.  In
order to show that the ``yes'' instances can be taken to have a
spectral gap which is at least inverse polynomial in the system size,
we would need a quantum version of
the Valiant-Vazirani theorem \cite{VV:86}, which would say that we
can always modify a \QMA\ problem to have a unique witness accepted
with high probability.

\section{Acknowledgements}
This work was partly done while three of the authors (D.~A., D.~G.,
and J.~K.) were visiting the Institute Henri Poincar\'e in Paris,
and we want to thank the IHP for its hospitality.  We wish to thank
Daniel Fisher, Matt Hastings, Lev Ioffe, Tobias Osborne, Oded Regev,
and Barbara Terhal for helpful discussions and comments and Oded for
help with the design of the state icons.


\newcommand{\etalchar}[1]{$^{#1}$}

\end{document}